\documentclass[11pt]{article}
\usepackage[margin=1in,bottom=1in]{geometry}
\usepackage[utf8]{inputenc}
\usepackage[T1]{fontenc}
\usepackage{libertinus}

\usepackage{amsthm}
\usepackage{amsmath}
\usepackage{amssymb}
\usepackage{stmaryrd}
\usepackage{mathrsfs} 
\usepackage{mathtools}
\usepackage{thm-restate}
\usepackage[colorlinks = true, linkcolor = blue, urlcolor = blue, citecolor = blue]{hyperref}
\usepackage[capitalize, nameinlink]{cleveref}
\usepackage{url}
\usepackage{todonotes}
\usepackage{soul}
\usepackage{relsize}
\usepackage{enumerate}
\usepackage[all,defaultlines=3]{nowidow}
\usepackage[mathlines]{lineno}
\usepackage{multicol}
\usepackage{enumitem}
\usepackage{xspace}
\usepackage{nicefrac}
\usepackage{wasysym}
\usepackage{textpos}

\usepackage{graphicx}
\usepackage[T1]{fontenc}
\usepackage{csquotes}
\usepackage{microtype}
\usepackage{booktabs}
\usepackage{multirow}
\usepackage{wrapfig}
\usepackage[margin,noinline,draft]{fixme}
\usepackage{cryptocode}
\usepackage{tikz}
\usepackage{circuitikz}
\usetikzlibrary{matrix}
\usetikzlibrary{graphs}
\usetikzlibrary{positioning, arrows.meta, calc}
\usepackage{tabularx}

\newcommand{\Z}{\mathbb{Z}}

\DeclareMathOperator{\poly}{poly}

\newcommand{\tup}[1]{\mathbf{#1}}

\renewcommand{\leq}{\leqslant}
\renewcommand{\geq}{\geqslant}
\renewcommand{\le}{\leqslant}
\renewcommand{\ge}{\geqslant}

\DeclarePairedDelimiter{\set}{\{}{\}}
\DeclarePairedDelimiter{\norm}{\lVert}{\rVert}

\newcommand{\hy}{\hbox{-}\nobreak\hskip0pt}

\newcommand{\NPc}{$\mathsf{NP}$\hy{}complete\xspace}
\newcommand{\NPh}{$\mathsf{NP}$\hy{}hard\xspace}
\newcommand{\FPT}{$\mathsf{FPT}$\xspace}
\newcommand{\Ptime}{$\mathsf{P}$\xspace}

\setlist[itemize]{topsep=8pt,itemsep=8pt,parsep=0pt} 
\setlist[enumerate]{topsep=8pt,itemsep=8pt,parsep=0pt} 

\newcommand{\aff}[1]{\textcolor{black!60}{\small{#1}}}

\newtheorem{theorem}{Theorem}
\newtheorem{corollary}[theorem]{Corollary}

\newtheorem{claim}[theorem]{Claim}

\begin{document}

\newcommand{\funding}{M.P. was supported by the project BOBR that is funded from the European Research Council (ERC) under the European Union’s Horizon 2020 research and innovation programme with grant agreement No. 948057.
K.P. was supported by the project GA24-11098S of the Czech Science Foundation at the beginning of the project, and by the funding received from the European Research Council (ERC) under the European Union’s Horizon 2020 research and innovation programme (grant agreement No 101002854) towards the end of the project.}

\title{On Integer Programs That Look Like Paths\footnote{\funding}}

\date{}

\author{}
 \author{
   Marcin Briański\\
   \aff{Jagiellonian University} \\
   \aff{marcin.brianski@doctoral.uj.edu.pl}
   \hspace*{-2.2em}
   \and
   Alexandra Lassota\\
   \aff{Eindhoven University of Technology} \\
   \aff{a.a.lassota@tue.nl}
   \and
   Kristýna Pekárková \\
   \aff{AGH University of Krakow} \\
   \aff{pekarkova@agh.edu.pl}
   \and
   Michał Pilipczuk \\
   \aff{University of Warsaw} \\
   \aff{michal.pilipczuk@mimuw.edu.pl}
   \and
   Janina Reuter\\
   \aff{Kiel University} \\
   \aff{janina.reuter@informatik.uni-kiel.de}
 }
\maketitle

\begin{abstract}
Solving integer programs of the form $\min\bigl\{\tup  x\bigm\vert A\tup  x =\tup b,\tup l\le \tup x\le \tup u,\tup x\in\Z^n\bigr\}$ is, in general, \NPh. Hence, great effort has been put into identifying subclasses of integer programs that are solvable in polynomial or \FPT time.  
A common scheme for many of these integer programs is a star-like structure of the constraint matrix. The arguably simplest form that is not a star is a path. 
We study integer programs where the constraint matrix $A$ has such a path-like structure: every non-zero coefficient appears in at most two consecutive constraints. We prove that even if all coefficients of $A$ are bounded by $8$, deciding the feasibility of such integer programs is \NPh via a reduction from $3$-SAT. Given the existence of efficient algorithms for integer programs with star-like structures and a closely related pattern where the sum of absolute values is column-wise bounded by 2 (hence, there are at most two non-zero entries per column of size at most 2), this hardness result is surprising.
\end{abstract}

 \begin{textblock}{20}(-1.75, 4.9)
 \includegraphics[width=40px]{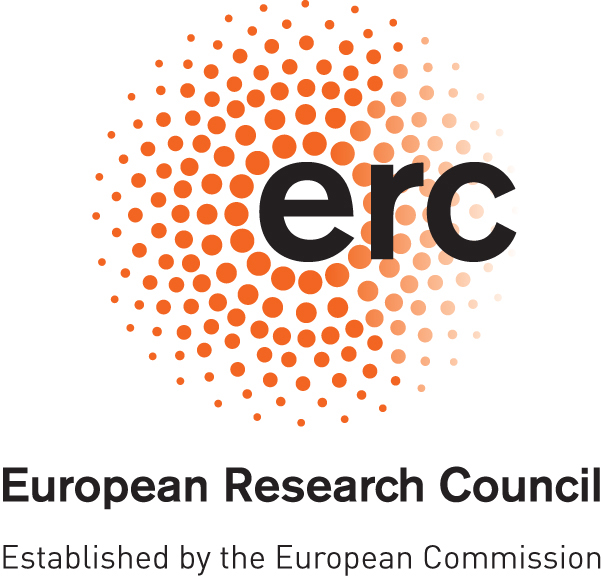}%
 \end{textblock}
 \begin{textblock}{20}(-1.75, 5.9)
 \includegraphics[width=40px]{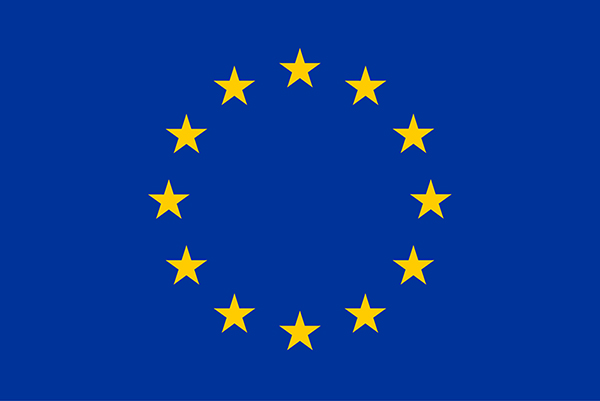}%
 \end{textblock}

\newpage
\clearpage
\setcounter{page}{1}

\section{Introduction}
We study integer programs (to which we refer to as \emph{IP}s) of the form
\begin{equation}
    \min\bigl\{\tup x\bigm\vert A\tup x=\tup b,\tup l\le \tup x\le \tup u,\tup x\in\Z^n\bigr\},
    \label[IP]{IP:general}\tag{G}
\end{equation}
with constraint matrix $A\in\Z^{m\times n}$, right-hand side $\tup b \in \Z^m$, lower and upper bounds $\tup l\in(\Z\cup\{-\infty\})^n$ 
and $\tup u\in(\Z\cup\{\infty\})^n$. Solving such IPs means to either decide that the system is infeasible, there exists an optimal feasible solution (and give the value of it), or it is unbounded and provides a solution with an unbounded direction of improvement. Note that restricting to equality constraints does not restrict the problem, as every inequality constraint can be turned to equality by introducing slack variables (which also preserve the special structures we discuss in this paper).

Integer programming is a powerful tool that has been applied to many combinatorial problems, such as scheduling and bin packing \cite{DBLP:journals/jacm/GoemansR20,DBLP:journals/mp/JansenKMR22}, graph problems \cite{DBLP:conf/isaac/FellowsLMRS08,DBLP:journals/dam/FialaGKKK18}, multichoice optimization \cite{ermolieva2023connections} and computational social choice \cite{bartholdi1989voting,DBLP:journals/teco/KnopKM20}, among others. Unfortunately, solving integer programs is, in general, \NPh.

This motivated a flourishing line of research that aims to understand the structural properties of matrices that allow the corresponding IPs to be solved efficiently, see, e.g.~\cite{DBLP:conf/isaac/FellowsLMRS08,DBLP:journals/jacm/GoemansR20,DBLP:journals/algorithmica/GrammNR03,DBLP:journals/mp/JansenKMR22,DBLP:journals/teco/KnopKM20}. Arguably, the most famous is Lenstra's 1983 algorithm, which introduced the first \FPT time algorithm for constraint matrices with few columns~\cite{DBLP:journals/mor/Lenstra83}. The body of literature for the over three decades-long research and the many structures, parameters, and applications is too vast to cover here; we thus refer to~\cite{GavenciakKK22} for an overview. 

One striking property for many of these IPs is the star-likeness of the constraint matrix, which is probably best exemplified by {\em{$n$-fold}} and {\em{$2$-stage stochastic}} integer programs. Here, an integer program is of $n$-fold form if its constraint matrix, after removing at most~$k$ rows ({\em{aka}} constraints), consists of $n$ independent blocks with at most $k$ rows each. The transpose of this matrix is called $2$-stage stochastic integer programs, i.e., after removal of at most $k$ columns ({\em{aka}} variables), they are decomposed into independent blocks with at most $k$ columns each. As proven in~\cite{DBLP:journals/mp/HemmeckeS03,DBLP:journals/mp/HemmeckeOR13}, the optimization problem for $n$-fold and $2$-stage stochastic integer programs can be solved in time $f(k,\Delta)\cdot |I|^{1+o(1)}$, where $f$ is a computable function, $\Delta$ is the largest absolute coefficient in the constraint matrix, and $|I|$ is the total bitsize of the instance. This running time is \FPT with parameters $k$ and $\Delta$.

The structure of $n$-fold and 2-stage stochastic integer programs can be generalized by allowing multiple stages of deleting rows or columns and decomposing into blocks. This gives rise to so-called {\em{tree-fold}} and {\em{multistage stochastic}} integer programs, which retain efficient solvability, see, e.g.,~\cite{EisenbrandHKKLO19,DBLP:conf/esa/CslovjecsekEPVW21}. The structure of tree-fold and multistage stochastic IPs can be perhaps best understood through decompositions of the associated graphs. For a matrix $A$, we define the {\em{constraint graph}} of $A$ as the graph on the rows (constraints) of $A$ in which two rows are adjacent if they simultaneously have non-zero entries in the same column (in other words, there is a variable appearing in both the corresponding constraints). The {\em{variable graph}} of $A$ is defined symmetrically for the columns (variables) of $A$. Then, the {\em{primal treedepth}} of $A$ is the treedepth of its variable graph, and the {\em{dual treedepth}} of $A$ is the treedepth of its constraint graph. Here, the treedepth of a graph is a parameter that measures star-likeness through the maximum depth of a procedure that alternately deletes a vertex and decomposes the graph into connected components, see for example~\cite[Chapter~6]{sparsity} for a broad introduction. As proven in~\cite{JansenK15,GanianO18,GanianOR17}, the optimization problem for integer programs can be solved in time $f(d,\Delta)\cdot |I|^{1+o(1)}$ for a computable function $f$, where $d$ is either the primal or the dual treedepth of the constraint matrix. We refer the reader to the article of Eisenbrand, Hunkenschr\"oder, Klein, Kouteck\'y, Levin and Onn~\cite{EisenbrandHKKLO19} for a comprehensive introduction to the theory of parameterized algorithms for block-structured integer programming.

It is, therefore, natural to ask to what extent the (primal or dual) treedepth of the constraint matrix captures the efficient solvability of integer programs. Several results in the literature suggest that integer programming quickly becomes \NPh once one relaxes the assumption of bounded treedepth. For instance, Ganian and Ordyniak~\cite{GanianO18} showed that the feasibility problem for integer programs of the form $\{\tup x\,\mid\,A\tup x\leq \tup b\}$ is $\mathsf{NP}$-hard even when all the entries of $A$ and $\tup b$ are in $\{-2,-1,0,1,2\}$ and the treewidth of the constraint graph of $A$ is at most $4$. Their reduction is based on an earlier work of Jansen and Kratsch~\cite{JansenK15}, who achieved a similar result, but without obtaining the bound on the entries. We note that Ganian and Ordyniak stated their result in terms of the treewidth of the {\em{incidence graph}} of $A$: the bipartite graph on rows and columns of $A$ where a row is adjacent to a column if the entry at their intersection is non-zero. As shown in~\cite{GanianO18}, the treewidth of the incidence graph of the constructed matrix is at most $3$, but one can also easily verify that the treewidth of the constraint graph is at most $4$ (bounding the treewidth/treedepth of the constraint graph implies a bound on the treewidth of the incidence graph, but not vice versa.) Similar lower bounds were also reported in the work of Ganian, Ordyniak, and Ramanujan~\cite{GanianOR17}. Later, Eiben, Ganian, Knop, Ordyniak, Pilipczuk, and Wrochna~\cite{EibenGKOPW19} showed that when the incidence graph is considered, even bounded treedepth is not enough: deciding feasibility of integer programs of the form $\{\tup x\,\mid\,A\tup x=0$, $\tup l\leq \tup x\leq \tup u\}$, is $\mathsf{NP}$-hard even when the entries of $A,\tup l,\tup u$ are in $\{-1,0,1\}$ and the incidence graph of $A$ has treedepth $5$. 
Klein, Polak, and Rohwedder~\cite{Klein0R23} show NP-hardness for constraint matrices, where a small $3 \times 5$ block matrix is repeated in a (lower-left) triangle, i.e., the same matrix is repeated over the diagonal and below, and entries in $\{-2,-1,0,1\}$.
Finally. Eiben, Ganian, Knop, and Ordyniak~\cite{EibenGKO18} showed a number of lower bounds showing that allowing $A$ to have arbitrarily large entries quickly leads to $\mathsf{NP}$-hardness, even if $A$ has bounded treedepth.

The goal of this work is to contribute to the understanding of the complexity of integer programming by considering matrices that have the simplest possible form while still having unbounded treedepth. Since length of the longest path (contained as a subgraph) is functionally equivalent to treedepth, the simplest possible graph with arbitrarily large treedepth is a path. Thus, our aim is to understand the complexity of integer programs whose constraint graph\footnote{One could also consider the setting when the variable graph is a path, but a quick examination of this setting reveals that it makes little sense, and in fact reduces to a subcase of the setting where the constraint graph is a path.} is a path. Equivalently, these are integer programs of the form~$\{\tup x\,\mid\,A\tup x=\tup b$, $\tup \ell\le \tup x\le \tup u\}$ where every variable can appear only in at most two constraints, which moreover must be consecutive. Up to permuting the columns, the constraint matrix $A$ must have the path-like structure as depicted in \cref{fig:matrixstructure}: the non-zero entries in each row may only appear in an interval, these intervals are naturally ordered from left to right, and only the intervals in consecutive rows may overlap.

\begin{figure}[!ht]
\centering
\resizebox{0.7\textwidth}{!}{%
\begin{circuitikz}
\tikzstyle{every node}=[font=\Huge]

\draw[fill=lightgray!40, line width=0.5pt] (5,18.75) rectangle (13.5,17.5);
\draw[fill=lightgray!40, line width=0.5pt] (7.5,17) rectangle (22.75,15.75);
\draw[fill=lightgray!40, line width=0.5pt] (15,15.25) rectangle (29.25,14);
\draw[fill=lightgray!40, line width=0.5pt] (31.25,11.5) rectangle (40,10.25);
\draw[fill=lightgray!40, line width=0.5pt] (35,9.75) rectangle (42.25,8.5);

\draw[line width=0.7pt] (5.25,19.5) to[short] (4.5,19.5);
\draw[line width=0.7pt] (4.5,19.5) to[short] (4.5,8);
\draw[line width=0.7pt] (4.5,8) to[short] (5.25,8);
\draw[line width=0.7pt] (42.75,8) to[short] (42.75,19.75);
\draw[line width=0.7pt] (42.75,19.75) to[short] (42,19.75);
\draw[line width=0.7pt] (42.75,8) to[short] (42,8);

\node [font=\Huge] at (27.5,13) {$\ddots$};
\end{circuitikz}
}
\caption{A schematic depiction of a path-like constraint matrix $A$.}
\label{fig:matrixstructure}
\end{figure}
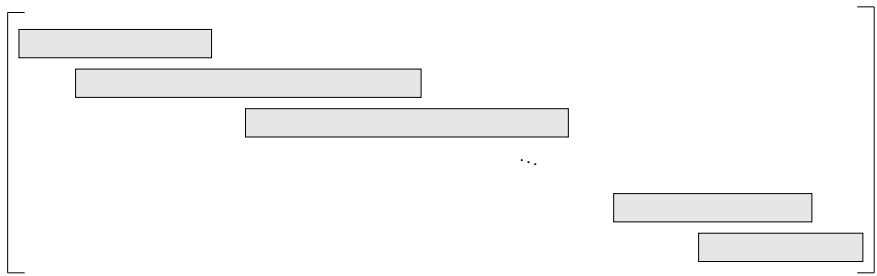

It is known that if the coefficients of the constraint matrix are in $\{-2,-1,0,1,2\}$ and the sum of absolute values is at most $2$ for each column, then the problem is polynomial time solvable~\cite{schrijver2003combinatorial}. Indeed, the corresponding problem called (generalized) matching problem is well-studied. It captures a variety of well-known problems in \Ptime such as minimum cost flow, minimum cost ($b$-)matching and certain graph factor problems~\cite{schrijver2003combinatorial}. These structures even remain \FPT time solvable parameterized by the number $p$ of additional columns~\cite{lassota2025parameterizedalgorithmsmatchinginteger}.

The case of \emph{two} non-zero entries per column seems to be indeed the turning point where analyzing the computational complexity becomes intriguing. Indeed, it is easy to see that the problem is \NPh when \emph{three} non-zero entries per column are allowed and entries are in $\{-1,0,1,2\}$. Consider the following reduction from subset sum where $n$ non-negative numbers $a_1, \ldots, a_n$ and a target $t\in\mathbb{Z}_{\geq 0}$ are given and the goal is to decide whether a subset of $\{a_1, \ldots, a_n\}$ sums up to $t$. Let $enc(x)$ be the binary encoding of a non-negative number $x$ and consider the matrices
\begin{align*}
    E = \begin{pmatrix}
        -1 & 2 & && \\
        &-1 & 2 & & \\
        && \ddots & \ddots &\\
        &&& -1 & 2
    \end{pmatrix}\qquad \text{ and }
    \qquad A = \begin{pmatrix}
        enc(a_1) &  \ldots & enc(a_n) \\
        E &&  \\
        & \ddots &  \\
        && E
    \end{pmatrix}.
\end{align*}
The first matrix, $E$, is called the \emph{encoding matrix} and is often used in lower bounds for IPs with a special structure of non-zero entries in the constraint matrix. The IP defined by the constraint matrix $A$, right-hand side $\tup b=(t, 0, \ldots, 0)$, and suitable upper and lower bounds $\tup l$ and $\tup u$ is feasible if and only if the corresponding subset sum instance is feasible. Clearly, the constraint matrix $A$ has at most three non-zero entries per column.

In stark contrast to the above-mentioned results for an even more restricted form in terms of the placement of non-zero coefficients, we show that the feasibility problem for integer programs with path-like constraint matrices is \NPh already when all entries are bounded by a constant.

\begin{theorem}\label{thm:lb}
    The problem of deciding whether a given integer program $\{\tup x\,\mid\,A\tup x=\tup b$, $\tup l\le \tup x\le \tup u\}$ has a solution is \NPh even when $A$ is a path-like matrix with all entries of absolute value at most $8$. Here, we assume that $\tup b$ is provided on input in binary.
\end{theorem}

We stress that the assumption that $\tup b$ is encoded in binary is essential. Indeed, since the incidence graphs of path-like matrices have treewidth at most~$2$, from the result of Ganian, Ordyniak, and Ramanujan~\cite[Theorem~13]{GanianOR17} it follows that the feasibility of integer programs of the form discussed in \cref{thm:lb} can be decided in time polynomial in $\|A\|_\infty+\|\tup b\|_{\infty}$ and the total bitsize of the input.

\cref{thm:lb} therefore provides a strong lower bound for integer programming and intuitively shows that any deviation from constraint matrices with bounded treedepth renders the problem intractable. The proof of the theorem is significantly more elaborate than the proofs of the lower bounds from~\cite{JansenK15,GanianO18,GanianOR17,EibenGKOPW19,EibenGKO18} mentioned above, which to a large extent relied on a direct encoding of {\sc{Subset Sum}} as an integer program with a simple structure. Our reduction is from {\sc{3-SAT}}. Before we prove this result in Section~\ref{sec:hardness}, we give an intuitive argument and sketch of the reduction at the beginning of the section. 

\section{Hardness}\label{sec:hardness}
We now prove that path-like integer programs are \NPh even when all coefficients of the constraint matrix are bounded by a constant. We prove this by reducing the {\sc{3-SAT}} problem to path-like IPs. 

The main idea is to design constraints such that each constraint $C_i$ in the matrix restricts the solution space of feasible assignments to be feasible for the first $i$ constraints. This way, starting with a constraint that encodes all possible variable assignments, we propagate and constrict the feasible assignments according to the clauses.
We construct the matrix in such a way that every variable occurs at most twice and in consecutive constraints. We say that the first constraint for a variable \emph{introduces} the variable, and that in a (potential) second constraint, the solution space of the variable is \emph{forwarded} or \emph{output}.
We use a limited number of different \emph{constraint types} to transform the solution space of a forwarded variable to a solution space for an introduced variable. Here, we consider the solution space as the set of feasible solutions for that variable considering the current and all previous constraints but ignoring all following constraints. The constraint types of the reduction are summarized in~\autoref{tab:path-like_operations}.

\begin{table}[h!]
    \centering
    \begin{tabular}{lll}
        \toprule
         Operation & Constraint \\
         \midrule
         Forward & $x-y=0$ \smallskip\\
         Check $x\geq a$& $x - y = a$ \smallskip\\
         Multiplication by $a\in \Z$ & $a x - y = 0$ \smallskip\\ 
         Division with remainder by $a\in \Z$ & $x - ay - r = 0$ and $0\leq r < a$ \smallskip\\ 
         Combination of two constraints  &  $\tup c_1^\intercal \tup x_1 = \tup b_1$ and $\tup c_2^\intercal \tup x_2 = \tup b_2$ with $\tup 0\leq \tup x_1 \leq \tup u$  \\
         with independent variables & combined to $ \tup c_1^\intercal \tup x_1 + a \cdot \tup c_2^\intercal\tup x_2 =\tup b_1 + a\tup b_2$ \\
         & with $a> \tup c_1^\intercal \tup u$
    \end{tabular}
    \caption{Types of constraints and the operation that is thereby performed on the solution space from an output variable $x$ and an introduced variable $y$ that are used in the proof of Theorem~\ref{thm:hardness}.}
    \label{tab:path-like_operations}
\end{table}

A cornerstone for the reduction is a number-theoretical identity that allows us to perform operations that check satisfaction of a clause while preserving the information stored in that variable. In particular, we need to look up the assignment of specific variables, check whether they yield a true literal in the clause, and then restore the original representation of the assignment. Considering only one clause, it would be relatively easy to repeatedly divide by $2$ with remainder and check whether the remainder implies a true literal. However, this procedure is in fact limited to one clause, as the information of the satisfying assignment are lost in the process.

\begin{theorem}\label{thm:hardness}
    The problem of deciding the existence of a vector \(\tup x \in \mathbb{Z}^n\)
    satisfying \(A\tup x = \tup b\), \(\tup l \leq \tup x \leq \tup u\) is \NPh even when
    \(A\) is a hollow path-like matrix with \(\norm{A}_{\infty} \leq 8\).
\end{theorem}
\begin{proof}
    We reduce from 3-SAT, which is known to be \NPc~\cite{DBLP:conf/stoc/Cook71}, i.e., we define a hollow path-like matrix $M$, lower and upper bounds $\tup l$ and $\tup u$, and a right-hand side $\tup b$ such that $M\tup x=\tup b$ has a solution $\tup l\leq \tup x \leq \tup u$ if and only if the 3-SAT formula is satisfiable.

    Let \(\varphi = \bigwedge_{i=1}^{m} C_i\) be a 3-SAT formula over the set of variables
    \(V = \set{v_1, v_2, \dots, v_n}\).
    We begin by checking whether the all-zero assignment \(\nu \colon V \ni v_j \mapsto 0 \in \set{0,1}\)
    is a satisfying assignment. If it does satisfy \(\varphi\), then the construction outputs a trivial
    instance where $A$ is the zero matrix, and $\tup b, \tup l, \tup u$ are zero vectors, i.e., \(\begin{bmatrix}0\end{bmatrix}\tup v = \begin{bmatrix}0\end{bmatrix}\), \(\tup l = \tup u = \begin{bmatrix}0\end{bmatrix}\).
    Henceforth, we assume that the constant 0 assignment is not a satisfying assignment of \(\varphi\).



   We first construct a constraint where the set of solutions contains all the $2^n -1$ possible variable assignments for the 3-SAT formula, except for the assignment mapping all variables to $0$, as follows:
    \begin{align*}
        x_{0,0} - \alpha &= 0,
    \end{align*}
    where the lower and upper bounds are $\alpha \in \{1,\ldots, 2^n-1\}$. The constraint is feasible for variable $x_{0,0} \in \mathcal A_0 =\{1,\ldots, 2^n-1\} = \{ [b_{n-1}b_{n-2}\ldots b_0]_2 \neq 0 \mid b_i \in\{0,1\}\}$, where $[b_{n-1}b_{n-2}\ldots b_0]_2$ represents the number $\sum_{i=0}b_i2^i$.

    The main part is to delete all representations of  non-satisfying assignments from the set. Each element of $\mathcal A_i$ represents a variable assignment that satisfies all clauses $C_j$ with $j\leq i$. Therefore, in order to generate the next set $\mathcal A_{i+1}$, we describe constraints that delete assignments that do not satisfy clause $C_{i+1}$. 
  
    Although it might seem natural to do the checking of the clause directly on the set $\mathcal A_{i}$, we need a trick in order to preserve the information about the assignment for checking further clauses.
    \begin{claim}\label{claim:xidentity}
    The following equations are true.
    \begin{enumerate}
        \item $x= \lfloor \lfloor x \cdot5^n/2^n \rfloor 2^n/5^n \rfloor +1$ for all $x\in\{1,\ldots 2^n-1\}$ 
        \item $\{x\cdot 5^n \mod 2^n\mid x\in\{1,\ldots 2^n-1\}\} = \{1,\ldots 2^n-1\}$
    \end{enumerate}
    \end{claim} 
    \begin{proof}
    Regarding the first equation, we get that the first floor rounds down by $0< (5^nx/2^n - \lfloor 5^nx/2^n \rfloor) \cdot 2^n/5^n < 1$ and the second floor rounds down by $2^ny/5^n - \lfloor 2^ny/5^n \rfloor < 1$ for $y = \lfloor 5^nx/2^n \rfloor$. In total, the rounding scrapes off more that $0$ but less than $2$, starts with an integer and ends with an integer. Thus, adding $1$ after the second rounding restores $x$. 
    The second equation follows directly as $2$ and $5$ are coprime.
    \end{proof}

    Roughly speaking, we will use the second equality in \autoref{claim:xidentity} to extract clause satisfiability from the remainder during the division by $2^n$ and then restore the original set of candidate assignments using the first identity of \autoref{claim:xidentity}.
    For any $x \in \mathcal A_i$ the represented assignment are the bits of $[b_{n-1}\ldots b_0]_2 = 5^nx\bmod 2^n$.   
    The following steps recreate the above identity while checking one clause for each application.\\
    
    \textbf{Preparation phase.} We first construct the set $\mathcal B_i = \{5^n \cdot a \mid a \in \mathcal A_i\}$. The following constraints transform variable $x_{i,0}$ with feasible assignments $\mathcal A_i$ to variable $y_{i,0}$ with feasible assignment $\mathcal B_i$
     \begin{align*}
        5x_{i,j} -  x_{i,j+1} &= 0 \quad \text{for all } 0\leq j < n\\
        x_{i,n} - y_{i,0} &= 0.
     \end{align*}
     Here, every constraint of the first type is a multiplication-by-$5$-type of constraint (see~\autoref{tab:path-like_operations}) and hence ensures that any feasible value for $x_{i,j+1}$ is five times any feasible value for $x_{i,j}$. The second constraint only adds syntactical sugar with $x_{i,n}=y_{i,0}$ to simplify variable names and readability in the next step.\\
     
    \textbf{Checking clause.} Following \autoref{claim:xidentity}, we use division by $2^n$ while checking the remainder for satisfaction of clause $C_i$. We use an additional variable $f_{i,j}$ to remember whether the clause is satisfied. This variable is increased if a literal of the clause is true for the assignment. If $f_{i,n}\neq 0$ after the division with remainder, then the clause is satisfied. 
    The goal is to achieve a set of feasible values for $(f_{i,n}, y_{i,n})$ of
    \begin{align*}
        \mathcal C_i = \{(f,\lfloor b / 2^n \rfloor) \mid b\in \mathcal{B}_i \text{ and }f = t_i(b)\},
    \end{align*}
    where $t_i(b)$ is the number of true literals in clause $C_i$ for assignment $[b_{n-1}\ldots b_0]_2 = b \bmod 2^n$.
    To be precise, in one case the constructed set contains additional elements where $f\leq t_i(b)$. However, this does not effect correctness as the following step will only check for values of $f>0$, hence for assignemnts, where there was at least one true literal found.
    
    For each SAT variable $v_j$, where $1\leq j \leq n$, we use the following types of constraints to transform any feasible value for $y_{i,0}$ from $\mathcal B_i$ to a corresponding pair $(f_{i,n}, y_{i,n})$ from $\mathcal C_i$:

    \begin{align*}
        f_{i,j} - f_{i,j+1} + 4 y_{i,j} - 8 y_{i,j+1} - 4\alpha_{i,j} &= 0 \quad \text{if }v_{j+1}\text{ is not in clause }C_i\\
        f_{i,j} - f_{i,j+1} + 4y_{i,j} - 8y_{i,j+1} - 3\alpha_{i,j} &= 0 \quad \text{if }v_{j+1}\text{ is in clause }C_i\\
        f_{i,j} - f_{i,j+1} + 4y_{i,j} - 8y_{i,j+1} - 5\alpha_{i,j} + \alpha'_{i,j} &= 0 \quad \text{if }\lnot v_{j+1}\text{ is in clause }C_i.\\
    \end{align*}
    We define variable bounds $\alpha_{i,j}, \alpha'_{i,j} \in \{0,1\}$. Further we define  $f_{i,j} \in \{0,1,2,3\}$ if all literals in clause $C_i$ are positive literals and  $f_{i,j} \in \{0,1,2\}$ if there is at least one negative literal. 

    The first constraint forwards the set of feasible values for $f_{i,j}=f_{i,j+1}$ and for every feasible value of $y_{i,j+1} = \lfloor y_{i,j}/2 \rfloor$. This is a combination of a division with remainder constraint type (for $y$ variables) and a forward constraint type (for $f$ variables), compare with \autoref{tab:path-like_operations}. As the variable $v_j$ is not in the clause, the remainder is irrelevant and gets absorbed by $\alpha_{i,j}$. The division by two, if independent from the context, could be achieved by 
    \begin{align*}
        y_{i,j} - 2y_{i,j+1} - \alpha_{i,j} = 0\quad\text{ and }\alpha_{i,j}\in\{0,1\}.
    \end{align*}
    However, we need to ensure that the feasible values for $f_{i,j}$ and $y_{i,j}$ do not interfere and create new feasible values. Hence, the multiplication of the above simple form of the division by $4 > f_{i,j+1} \in \{0,1,2,3\}$ makes sure that the feasible values transform independently as any value of $f_{i,j+1}$ can not compensate for an unwanted change in $y_{i,j+1}$.

    In the second and third type of constraint, the feasible values for $f_{i,j}$ are incremented in $f_{i;j+1}$ depending on the remainder from the division. As before, this is similar to a division with remainder constraint type and a forward constraint type. However, the variables should not evolve independently and instead the solution space for $f_{i,j+1}$ depends on the remainder of the division performed on $y_{i,j}$. 
    In the second constraint it is incremented if the remainder is $1$ and in the third constraint it \emph{could or could not} be incremented if the remainder is $0$. 
    
    In the case of the second constraint type, this is achieved from a coefficient of $3$ for $\alpha_{i,j}$ instead of $4$ (compared to the first type of constraint). 
    Then, in case of a remainder of $1$ (representing a true literal), the variable $f_{i,j+1}$ has to make up the difference of $1$. 
    
    In the third constraint type, if the remainder is $1$ and hence the literal is false, then the  feasible values either set $\alpha_{i,j} = \alpha'_{i,j} = 1$ and forwards the value of $f_{i,j+1} = f_{i,j}$ or set $\alpha_{i,j} = 1$, $\alpha'_{i,j} = 0$ and forwards the value of $f_{i,j+1} = f_{i,j} -1$, which is only possible if $f_{i,j}>0$. Hence, this does not transform an invalid $f_{i,j}$ to a false positive $f_{i,j+1}>0$ and similarly, any positive $f_{i,j}$ is forwarded. In the case that the remainder is $0$ and hence the literal is true, there are again two possible transformations. Either we set $\alpha_{i,j} = \alpha'_{i,j} = 0$ and $f_{i,j+1} = f_{i,j}$ or we set $\alpha_{i,j} =0$, $\alpha'_{i,j} = 1$, and $f_{i,j+1} = f_{i,j}+1$. In other words, if this literal is true and hence the clause is satisfied, the set of feasible solutions for $f_{i,j+1}$ contains also a value of $f_{i,j}+1>0$. The additional possibility of an unchanged $f_{i,j+1} = f_{i,j}$ does not matter as we eventually only require one feasible version of this assignment, when testing in the next step. Hence, this derivation from the defined set $\mathcal C_i$ have no implications on the reduction.
    Moreover, note that if this constraint is present $f_{i,j}, f_{i,j+1} \in \{0,1,2\}$ and thus $-f_{i,j+1} - 5\alpha_{i,j} > -8$ to ensure correct division by $2$. This does not alter the set of feasible assignments with $f_{i,j+1}>0$, since $f_{i,j}=f_{i,j+1}$ is feasible with $\alpha'_{i,j}=0$.\\
    
    \textbf{Delete not satisfied clauses.} Next, we delete any assignment that does not satisfy the clause $C_i$. We add feasible values $(f_{i,n}, y_{i,n})$ from $\mathcal C_i$ to the next set $\mathcal{D}_i$ iff $f_{i,n}>0$. We want to achieve a set of feasible solutions for $z_{i,0}$ of $\mathcal D_i = \{c \mid (f,c) \in \mathcal C_i \text{ with }f>0\}$. We can achieve this with the  constraint
    \begin{align*}
        f_{i,n} + 4y_{i,n} - 4 z_{i,0} - \beta_i = 1,
    \end{align*}
    where $\beta_i \in \{0,1,2\}$. Since $\beta_{i} < 3$, this constraint is feasible iff $f_{i,n}>0$ and then $y_{i,n} = z_{i,0}$. This is again a combination of constraint types with a checking type for $f_{i,n}$ and a forward type for $y_{i,n}$ to $z_{i,0}$. Also note that this transformation eliminates any derivation from $\mathcal C_i$ in the case of negated variables present in the clause as explained in the previous step.\\

    \textbf{Restore representation.} For the current clause $C_i$, it remains to restore the original representation using the identity from \autoref{claim:xidentity}. Hence, we transform any feasible value for $z_{i,0}$ from $\mathcal D_i$ to a value for $x_{i+1,0}$ from
    \begin{align*}
        \mathcal A_{i+1} = \{\lfloor d\cdot 2^n/5^n\rfloor +1 \mid d \in \mathcal{D}_i\}.
    \end{align*}
    The transformation uses similar multiplication and division constraint types as before, i.e.
    \begin{align*}
        2z_{i,j} - z_{i,j+1} &= 0 \quad \text{for all } 0\leq j < n\\
        z_{i,n+j} - 5z_{i,n+j+1} - \gamma_{i,j} &= 0  \quad \text{for all } 0\leq j < n\\
        z_{i,2n} - x_{i+1,0} &= -1,
    \end{align*}
    with variable bounds $\gamma_{i,j} \in\{0,1,2,3,4\}$.\smallskip

    Adding the constraints for \emph{preparations phase}, \emph{checking clause}, \emph{delete not satisfied}, and \emph{restore representation} for all clauses $C_1, \ldots, C_m$ results in a  constraint matrix $M$ with a hollow path-like structure such that the IP is feasible if and only if the 3-SAT formula is satisfiable. To be precise, \emph{restore representation} is not required for the last clause $C_m$ as it does not impact feasibility of the IP. 
    
    The set of feasible solutions for the IP at position $x_{0,0} = x_{1,0} = \ldots = x_{m,0}$ is exactly the set $\mathcal A_m$. Solutions of the 3-SAT formula are given by $5^n\cdot x_{m,0} \mod 2^n$ for any $x_{m,0}\in \mathcal A_m$. In other words, the set
    \begin{align*}\mathcal S=\{5^n\cdot a \mod 2^n \mid a\in \mathcal A_m\}.
    \end{align*}
    consists of all satisfying assignments $\nu \colon V \ni v_j \mapsto b_j \in \set{0,1}$ of $\varphi$ for $[b_{n-1}\ldots b_0]_2 \in \mathcal S$ and any feasible solution to the hollow path-like IP sets variable $x_{m,0}$ to an element of $\mathcal A_m$. 

    Hence, this problem is hard.    
\end{proof}

We use the construction from the above proof to additionally derive a fine-grained lower on the running time to solve path-like IPs assuming the Exponential Time Hypothesis (ETH). 

\begin{corollary}
    Assuming the ETH, the problem of deciding the existence of a vector \(\tup x \in \mathbb{Z}^n\) satisfying \(A\tup x = \tup b\), \(\tup l \leq \tup x \leq \tup u\) cannot be solved in time $2^{\delta \sqrt{n}} \cdot \poly(|I|) $ for some $\delta>0$, where $|I|$ denotes the encoding length, even when \(A\) is a hollow path-like matrix with \(\norm{A}_{\infty} \leq 8\).
\end{corollary}

\begin{proof}
    Consider a 3-SAT formula $\varphi$ over $\Tilde{n}$ variables. Using sparsification, assume that $\varphi$ has $O(\Tilde{n})$ clauses. The ETH states that there exists $\Tilde{\delta} >0$ such that 3-SAT can not be solved in time $2^{\Tilde{\delta} \Tilde{n}}$.

    The transformation in the proof of \autoref{thm:hardness} gives us an equivalent path-like IP with the largest entry in the coefficient matrix of $\norm{A}_\infty \leq 8$ and $\norm{\tup l}_\infty <\norm{\tup u}_\infty < 2^n$. The complexity of the transformation is dominated by the variables $x_{i,j}$, $y_{i,j}$, $z_{i,j}$, $z_{i,\Tilde{n} + j}$, $f_{i,j}$, $\alpha_{i,j}$, $\alpha'_{i,j}$, $\gamma_{i,j}$, where $i$ iterates over the clauses and $j$ over the variables (and possibly one more in both cases) of $\varphi$. Hence, the number of variables in the path-like IP is bounded by $n \in O(\Tilde{n}^2)$. As every variable of the path-like IP occurs in at most two constraints, this bound carries over to the number of constraints.

    Hence, hollow path-like IPs can not be solved in time 
    \begin{align*}
        2^{\delta \sqrt{n}} \cdot \poly(|I|) = 2^{\Tilde{\delta} \Tilde{n}}
    \end{align*}
    for some $\delta > 0$ even when the coefficient matrix has entries bounded by $8$.
\end{proof}

\section{Summary and Open Problems}\label{sec:SOP}
We showed that the feasibility problem for integer programs with path-like constraint matrices is \NPh already when all entries of the constraint matrix are bounded by $8$. We prove this result by a reduction from $3$-SAT. 
The reduction builds a path of constraints to repeatedly filter the set of all $3$-SAT assignments for those satisfying a clause from the formula until only assignments satisfying all clauses are left. The coefficient $8$ comes as a result of handling two types of information simultaneously in one constraint. This requires coefficients for one flow of information to be multiplied by the largest value for the other information in order prevent interaction between the variables of the two and their respective information held by the variables. Hence, the largest coefficient of our reduction is the largest value for one information flow $4$ multiplied by the largest coefficient $2$ in the other information flow. 
Thus, solving such IPs is harder than than solving IPs where the coefficients are bounded by $1$, which is known to be polynomial-time solvable (or more precisely, any IP with constraint matrix where the sum ob absolute values per column is bounded by $2$). 

This leaves two questions: First, what is the complexity of integer programs with path-like constraint matrices with coefficients smaller than $8$. Second, are integer programs with path-like constraint matrices but without upper bounds on the variables, that is, ILPs of the form $\{\tup x\,\mid\,A\tup x=\tup b$, $ \tup x\ge 0\}$, \NPh as well.

\bibliography{ref}
\bibliographystyle{plain}

\end{document}